\documentclass[a4paper]{article}
\usepackage{amsmath}
\usepackage{amscd}
\usepackage{amsthm}
\usepackage[mathscr]{eucal}
\usepackage{amssymb} 
\usepackage{latexsym}
\theoremstyle{definition}
\newtheorem{theorem}{Theorem}[section]
\newtheorem{main}{Main theorem}
\newtheorem{lemma}[theorem]{Lemma}
\newtheorem{proposition}[theorem]{Proposition}
\newtheorem{corollary}[theorem]{Corollary}
\newtheorem{remark}[theorem]{Remark}
\newtheorem{definition}[theorem]{Definition}

\begin{document}
\title{Notes on Feynman path integral-like methods of quantization
on Riemannian manifolds}
\author{Yoshihisa Miyanishi}
%\address{Center for the Study of Finance and Insurance, Osaka University
%Machikaneyamacho 1-3, Toyonakashi, Japan}
%\ead{miyanishi@sigmath.es.osaka-u.ac.jp}
\maketitle
%%%%%%%% abstract@%%%%%%%%%%%%%%%%%%%%%%%%%%%%%%%%%
\begin{abstract}
We propose an alternative method for Feynman path integrals on compact Riemannian manifolds. 
Our method employs the action integral $S(t, x, y)$ along the shortest path between two points.  
The corresponding oscillatory integral operator is defined by  
\begin{align*}
U_{\chi} (t)f(x)\equiv \frac{1}{(2\pi i)^{n/2}}\int_{M} \chi(d(x,y)) 
\sqrt{V(t, x, y)} e^{iS(t,\ x,\ y)}f(y)\; dy.  
\end{align*}
where $\chi(d(x,y))$ is the bump cut-off function with small compact support 
and $V(t, x, y)$ denotes van Vleck determinant. 
In the case of  rank $1$ locally symmetric Riemannian manifolds, 
we prove the strong convergence of time slicing products $\lim\limits_{N\rightarrow \infty}\{U(t/N)\}^{N}$ 
for low energy functions. 
Moreover, the strong limit includes Dewitt curvature $R/6$, 
where $R$ denotes the scalar curvature of a Riemannian manifold. 
This is an alternative rigorous formulation for Feynman path integrals on Riemannian manifolds.
\end{abstract}
%%%%%%% contents %%%%%%%%%%%%%%%%%%%%%%%%%%%%%%%%%%
%%%%%%% introduction %%%%%%%%%%%%%%%%%%%%%%%%%%%%%%%
\section{Introduction} 
Let $(M, g)$ be a compact, oriented, smooth Riemannian $n$-dimensional manifold without boundary. 
We know that the injective radius of a compact manifold $\underbar{\it d}$ is always 
finite and positive \cite{Be}. 
It is also known that the geodesic distance function $d(x, y)$ is $C^{\infty}$ in a neighborhood 
of  $(x, y)\in M\times M$ if and only if  $x$ and $y$ are not conjugate points along 
this minimizing geodesic \cite{K-N}. 
Thus the smooth geodesic action $S(t, x, y)$ is represented 
as an integral over time, taken along the geodesic path between the initial time 
and the final time of the development of the system: 
$$ S(t,\ x,\ y)=\int_0^t \frac{1}{2}g_{x(t)}(\dot{x}(t), \dot{x}(t))\; dt 
=\frac{|d(x, y)|^2}{2t} \quad \mbox{for}\ d(x, y)<\underbar{\it d},$$ 
and the path-density between two points {\cite{Wa}} is given by van Vleck determinant: 
$$V(t,x,y)=g^{-1/2}(x)g^{-1/2}(y) \det _{ij}
\left(\frac{\partial^2 S(t,\ x,\ y)}
{{\partial {x_i}}{\partial {y_j}}}\right).
$$
%Here, $g(x)={\det g_{ij}}(x)$ is the absolute value of the determinant 
%of the matrix representation of the metric tensor on the manifold.
We define a reasonable candidate for the short-time 
quantum propagators associated to $S$ and $V$ by a family of oscillatory integral operators.   
\par
\begin{definition}[Shortest path approximations on $M$]
%Let $L(x,\dot{x})=\frac{|\dot{x}|^2}{2}
%=\frac{1}{2}g_{x(t)}(\dot{x}(t), \dot{x}(t))$ 
%be the Lagrangian function. 
The shortest path approximation 
$U_{\chi}(t)$ on $M$ is defined by 
\begin{align*}
U_{\chi} (t)f(x)\equiv \frac{1}{(2\pi i)^{n/2}}\int_{M} \chi(d(x,y)) 
\sqrt{V(t, x, y)} e^{iS(t,\ x,\ y)}f(y)\; dy,  
\end{align*}
where $\chi(d(x,y))$ is the bump function with compact support contained 
in $d(x, y)<\underbar{\it d}$.  
\end{definition}
It is emphasized that an approximate ``local" parametrix of  $e^{\frac{it\triangle}{2}}$ 
is based on the above form for extremely short time interval (See e.g. \cite{H-T-W}). 
For the purpose of local to global in time consistency, we also introduce   
Feynman path integral-like methods of quantization by the limit of time slicing products:  
\begin{equation}
\lim_{N\rightarrow \infty} \{U_{\chi}(t/N)\}^N. 
\end{equation}
%As will be explained below this is the analogue of Feynman path integrals. 
%\par   
This paper aims to give the meaning of this limit and to evaluate the limit if it converges. 
However, we have not accomplished this task for general compact manifolds. So, 
as the first observation we restrict ourselves to the class of rank $1$ locally symmetric 
Riemannian manifolds, including the $n$-dimensional sphere $S_n$ 
and hyperbolic manifolds etc. (See e.g. \cite{SH} and \S4). 
Our main result is the following: 
%%%%%%%%%%%%%%%%%%%%%%%%%%%%%%%%%%%%%%%%%%%%%
\begin{main}[Time slicing products and the strong limit (See \S4)]
Let $(M, g)$ be a compact, oriented, rank 1 locally symmetric Riemannian manifold. 
For $f(x) \in  L^2(M)$, 
$$
s\hspace{-1.5mm}\lim_{N\rightarrow \infty} \{U_{\chi}(t/N)\}^N \rho(N) f(x) 
=e^{\frac{it}{2}(\triangle-\frac{R}{6})} f(x) \quad\mbox{in}\ L^2. 
$$ 
where $R$ means the scalar curvature and 
$\rho(N)$ is a spectral measure defined by the spectral theorem : 
$-\triangle=\int_{0}^{\infty}E\ d\rho(E)$. 
\end{main} 
%%%%%%%%%%%%%%%%%%%%%%%%%%%%%%%%%%%%%%%%%%%%%%%%%%
If $f(x)$ is a low energy function (i.e. a finite sum of eigenfunctions of $-\triangle$), 
the covergence of time slicing products is given without spectral projectors: 
%%%%%%%%%%%%%%%%%%%%%%%%%%%%%%%%%%%%%%%%%%%%%%%%
\begin{corollary}
Under the hypothesis of Main Theorem,  
if $f(x)=\sum\limits_{{\rm finite}} u_j(x)$ is a finite sum of 
Laplace eigenfunctions, then  
$$
s\hspace{-1.5mm}\lim_{N\rightarrow \infty} \{U_{\chi}(t/N)\}^N f(x) 
=e^{\frac{it}{2}(\triangle-\frac{R}{6})} f(x) \quad\mbox{in}\ L^2. 
$$ 
\end{corollary}
%%%%%%%%%%%%%%%%%%%%%%%%%%%%%%%%%%%%%%%%%%%%%%%%%%%%%%%%%%%%
This is an analogous result for Feynman path integral proposed 
by means of finite dimensional approximations and Trotter type time 
slicing products (See e.g. \cite{Fu 1},\cite{Fu 2}, \cite{Fu 3}, 
\cite{Fu-Tu}, \cite{I-W}, \cite{Ino}, \cite{Int}, \cite{Ki-KH}, 
\cite{KN}, \cite{Ya 1}). In these papers, the stationary action 
trajectories are finite for fixed time $t>0$ and $x, y\in \mathbf{R}^n$,  
and the kernel $E(t,x,y)$ of $e^{it( \frac{-\triangle}{2}+V(x))}$ 
are bounded smooth for small $t\not=0$. Thus time slicing products 
converge without spectral projectors.  
\par
On compact manifolds, however, infinite many 
action paths exist, even if time $t>0$ is fixed. 
To clarify the meaning of spectral projectors,
we consider the quantum evolution $e^{\frac{it}{2}(\triangle-\frac{R}{6})}$ on $M$. 
By Stone's theorem {(See e.g. \cite{RS})}, $e^{\frac{it}{2}(\triangle-\frac{R}{6})}$ 
are unitary operators and the kernel are given by  
$$
E(t,x,y)=\sum_{E_j}e^{\frac{-it{E_j}}{2}}\overline{u_j(x)}u_j(y)
$$
where $\{u_j(x)\}$ is eigenfunction expansion of $-\triangle+\frac{R}{6}$ 
and $E_j$ are eigenvalues.
The behavior of $E(t,x,y)$ is quite singular 
(See e.g. \cite{Ka}, \cite{Ni}, \cite{Ta 1}, \cite{Ta 2}, \cite{Ya 2}). 
Neverthless, when we sum a finite number of terms in $E$, $E_{finite}(t,x,y)$ are smooth 
and  we may intuitively choose classical shortest paths for low energy $E$. 
%%%%%%%%%%%%%%%%%%%%%%%%%%%%%%%%%%%%%%%%%%%%%%%%%%%%%%%%%%%%
Accordingly we may define the heuristic approximation for 
Feynman path integration by $\{U_{\chi}(t/N)\}^N 
\rho({N^{1/\alpha-\varepsilon}})$. 
Indeed, uniform convergences are proven in \S 3:
\setcounter{section}{3}
\setcounter{theorem}{0}
\begin{proposition}
For $\alpha=2+\frac{1}{2}[\frac{n+2}{2}]$ and small $\varepsilon>0$,  
$$
\lim_{N\rightarrow \infty} \Vert [\{U_{\chi}(t/N)\}^N -e^{\frac{it}{2}(\triangle-\frac{R}{6})}] \rho(N^{1/\alpha-\varepsilon}) \Vert_{L^2}=0.  
$$ 
\end{proposition} 
In \S 4 the strong convergence is assured by $L^2$ estimates, 
replacing $\rho(N^{1/\alpha-\varepsilon})$ with $\rho(N)$. 
Also notice \cite{Mi} that this convergence is not uniform without spectral projectors. 
\par 
Another way to understand the low energy is WKB method  
in which well-known $h$-small semiclassical calculus gives  
the low energy good parametrices of Schr\H{o}dinger operators 
(See for instance \cite[p.581]{BGT}, \cite{Ro}). 
So the low energy approximation is just a rewording WKB method 
in less $h$-small terminology. 
\par 
We end this introduction with some reasons why the amplitude is considered as van Vleck determinant.  
In physics literature, Feynman is saying in his book \cite{FH} that each trajectory 
contributes to the total amplitude to go from $a$ to $b$ and 
that they contribute equal to the amplitude, but contribute 
at different phases (i.e. The amplitude is constant in the original idea). 
Some researchers however propose  that the amplitude should be given by  
the density of trajectories \cite{Schulman}. Since 
the limit (1) converges to Schr\H{o}dinger operator $e^{it( \frac{-\triangle}{2}+V(x))}$ in the Euclidean case 
and the amplitude of its kernel reduces to the 
expression with van Vleck determinant (See e.g. \cite{Fu 3}). 
One more reason to consider van Vleck determinants is the accuracy of convergence (See \S 3). 
Thus it is natural to employ the amplitude as the density of trajectories.  
For the heat semi-group case, 
this ambiguity in the path integral is discussed from a strictly mathematical 
viewpoint \cite{And-Dri}. 
\setcounter{section}{1}
\setcounter{theorem}{2}
%%%%%%%%%%%%%%%%%%%%%%%% outline of proofs %%%%%%%%%%%%%%%%%%%%%%%%%%%%%%%%%%%%%%%%%%%
\section{Preliminaries}
\par
We start out by giving the properties of geodesic flows, van Vleck determinants and 
stationary phase methods on $M$. 
Throughout \S 2, we only assume $(M, g)$ is a compact, oriented, 
smooth Riemannian $n$-dimensional manifold without boundary. 
%%%%%%%%%%%%%%%%%%%%%%%%%%%%%%%%%%%%%%%%%%%%%%%%%%%%%%%
\par
Recall that the tangent space $T^*M$ admits a symplectic structure, which can be expressed locally as 
$\sum\limits_{i} dx^{i}\wedge dp_i$. Consider the smooth function defined on $T^{*} M$ by
$$
H(x, p)=\frac{1}{2} g^{ij}p_i p_j \equiv \sum_{i, j}\frac{1}{2} g^{ij}p_i p_j. 
$$
Here we often use Einstein summation convention as above, that is, 
in any expression containing subscripted variables appearing twice
(and only twice) in any term, the subscripted variables are assumed to
be summed over. The Hamitonian vector field of $H$ is 
$
X_H =\frac{\partial H}{\partial p_i}\frac{\partial}{\partial x^{i}}
-\frac{\partial H}{\partial x_i}\frac{\partial}{\partial p^{i}}
$
and its exponential map $\exp tX_{H}: T^{*}M\rightarrow T^{*}M$ is called the geodesic flow. 
It is well-known \cite{A-M} that the Legendre transform gives 
$p_i(t)=g_{ij}(t) {v^j}(t)\equiv g_{ij}(t) \dot{x^j}(t)$ for any geodesic $x(t)$ 
and the action integral is denoted by: 
$$ S(t,\ x,\ y)=\int_0^t \frac{1}{2}g_{x(t)}(\dot{x}(t), \dot{x}(t))\; dt 
=\frac{|d(x, y)|^2}{2t}, \quad \mbox{for}\ d(x, y)<\underbar{\it d}. $$ 
%%%%%%%%%%%%%%%%%%%%%%%%%%%%%%%%%%%%%%%%%
\par 
We also recall the definition of Riemann normal coordinates (See e.g. \cite{K-N} for details). 
To define a system of Riemann normal coordinates, one needs to pick a point $P$ 
on the manifold which will serve as origin and a basis for the tangent space at $P$.  
To any $n$-tuplet of real numbers $(x_1, \cdots, x_n)$, we shall assign a point $Q$ of 
the manifold by the following procedure: 
\par
Let $v$ be the vector whose components with respect to the basis chosen 
for the tangent space at $P$ are $x_1, \cdots, x_n$. 
There exists a unique affinely-parameterized geodesic $x(t)$ 
such that $x(0)=P$ and $[dx(t)/dt]_{t=0}=v$. Set $Q=x(1)$. Then $Q$ is defined 
to be the point whose Riemann normal coordinates are $(x_1, \cdots, x_n)$ for $d(P, Q)<\underbar{\it d}$.
Riemann normal coordinates enjoy several important properties: 
\par
1. The connection coefficients $\Gamma^{\alpha}_{\beta, \gamma}$ 
vanish at the origin of Riemannian normal coordinates.
\par
2. Covariant derivatives reduce to partial derivatives at the origin of Riemann normal coordinates.
\par
3. The partial derivatives of the components of the connection evaluated at the origin of Riemann normal coordinates equals the components of the curvature tensor.
\par
Under these circumstances, we consider the space initial data of the geodesic initial data problem 
on the tangent bundle $TM$. 
Equivalently the cotangent bundle $T^{*}M$ due to the existence of metric and Legendre transforms. 
Trading the dependence on initial momenta by the dependence on the final positions defines a 
map $T^*M\ \rightarrow\ M\times M$ whose Jacobian as will be shown below is the van Vleck 
determinant. From the derivation it will be explicit that is the inverse of Jacobian of the geodesic 
exponential map. 
%%%%%%%%%%%%%%%%%%%%%%%%%%%%%%%%%%%%%%%%%%%%%%%%%%%%%%%%%%%%%%%%%%%%%%%%%%%%%%
\begin{lemma}
Let $d(x, y)<\underbar{\it d}$. For Riemann normal coordinates at $y$,  we have 
$$V(t, x, y)=t^{-n}  \{ \det(g_{i j}(x)) \}^{-1/2}.$$
\end{lemma}
\begin{proof}
Remarking that the time $t$ means the scaling of $S(t, x, y)$,  
we may only prove the theorem for the case of $t=1$.  
The Liouville measure on $T^*M$ is given by : 
$$
d\mu_L=\wedge dx^i(0) \wedge dp_i(0). 
$$
The Hamilton-Jacobi function is equal to the action integral along the geodesic between the 
initial point and final points: 
$$
S(1, x, y)=S(1, x(1), x(0))=\frac{1}{2}\int_{x(0), t=0}^{x(1), t=1} dt \sum_{i, j}g_{ij}(x(t))\dot{x}^i(t)\dot{x}^j(t).
$$
$S(t, x, y)$ satisfies the Hamilton-Jacobi equation: 
$$
\frac{1}{2}\sum_{i, j}g^{ij}(x(0))\frac{\partial S}{\partial x^i(0)}\; \frac{\partial S}{\partial x^j(0)}=E, 
$$
where 
$
E=\frac{1}{2}\sum_{i, j}g_{ij}(x(t))\dot{p}^i(t)\dot{p}^j(t)
$
is the invariant energy. 
The initial momentum along the geodesic can be derived from Hamilton-Jacobi function by: 
$$
p_i(0)=\frac{\partial S}{\partial x^i(0)}. 
$$
Thus the transformed Liouville measure on $M\times M$ is given by: 
\begin{align*}
\wedge_i dx^i(0) \wedge_i dp_i(0)&=\det (\frac{\partial p(0)}{\partial x(1)}) \wedge_i dx^i(0) \wedge_i dx^i (1) \\
&=\frac{\det(\frac{\partial^2 S}{\partial x^i(0) \partial x^j(1)})}{\sqrt{\det g(x(0)) \det g(x(1))}} 
\sqrt{\det g(x(0))} \wedge_i dx^i(0) \sqrt{\det g(x(1))}  \wedge_i  dx^i (1) \\
&=V(1, x(1), x(0)) \sqrt{\det g(x(0))} \wedge_i dx^i(0) \sqrt{\det g(x(1))}  \wedge_i  dx^i (1) 
\end{align*}
This leads to the alternative expression of  van Vleck determinant:  
$$V(1, x, y)={\sqrt{\det g(x(0))}}^{\; -1} {\sqrt{\det g(x(1))}}^{\; -1}  
\det(\frac{\partial p(0)}{\partial x(1)}). $$ 
Now by definition: 
$$
x(1)=\exp^{1}_{x(0)}(v(0))
$$
where the velocity: $v(0) \in T_{x(0)}$ is given by the Legendre transform: 
$$
v^{i}(0)=\sum_{j} g^{ij}(x(0)) p_j(0). 
$$
So we have  
\begin{align*}
 \sqrt{\det g (x(1))}&=\det_g (d \exp^{1}_{x(0)}) \\
&={\sqrt{\det g(x(0))}}^{\; -1} {\sqrt{\det g(x(1))}} 
\det(\frac{\partial x(1)}{\partial v(0)}) \\
&={\sqrt{\det g(x(0))}}^{\; -1} {\sqrt{\det g(x(1))}}\det g(x(0)) \det(\frac{\partial x(1)}{\partial p(0)}) \\
&=V(1, x(1), x(0))^{-1}.
\end{align*}
\end{proof}
Next, we provide calculations of the kernel of $U_{\chi}(t)$. In the polar normal coordinates around $y$, the Laplacian looks like 
$$
\triangle f(x) =\frac{\partial ^2 f}{\partial r^2} (x)+H(x, y)\frac{\partial f}{\partial r}(x) +\frac{1}{r^2} \triangle_{S} f(x)
$$
where $r=d(y, x)$, $S$ is the geodesic sphere of radius $r$ centered at $y$, $\triangle_S$ is 
the Laplacian to $S$ and $H(x, y)$ is the total mean curvature at $x$ of $S$ (See e.g. \cite{K-N}). 
\par
Letting  $\widehat K(t, x, y)=\sqrt{V(t, x, y)} e^{iS(t,\ x,\ y)}
=a(t, r) \exp\left(\frac{ir^2}{2t}\right)$,  
\begin{align*} 
\triangle \widehat K(t, x, y)
&=\frac{\partial ^2 }{\partial r^2}\{{a(t, r)} e^{i\frac{r^2}{2t}}\} 
+H(x, y)\frac{\partial }{\partial r}\{{a(t, r)} e^{i\frac{r^2}{2t}}\} \\
&=\left\{ \frac{\partial ^2 a}{\partial r^2}{(t, r)} +\frac{2ir}{t}\frac{\partial a}{\partial r}{(t, r)}
   -\frac{r^2a(t, r)}{t^2}+\frac{ia(t, r)}{t} +H(x, y) \frac{\partial a}{\partial r}{(t, r)}+\frac{ir a(t, r) H(x, y)}{t} 
\right\}e^{i\frac{r^2}{2t}}, \\ 
i\frac{\partial}{\partial t}\widehat K(t,x,y)&=i \frac{\partial}{\partial t}\{{a(t, r)} e^{i\frac{r^2}{2t}}\}  
=\left\{ i \frac{\partial a}{\partial t}{(t, r)} +\frac{r^2 a(t, r)}{2t^2} \right\}e^{i\frac{r^2}{2t}}.
\end{align*}
Summarizing the calculations, we have
\begin{align*}
\left\{ i\frac{\partial}{\partial t}+\frac{1}{2}\triangle_x \right\} \widehat K(t, x, y)
&=\left\{ \frac{1}{2}\frac{\partial ^2 a}{\partial r^2}{(t, r)} +\frac{H(x, y)}{2} \frac{\partial a}{\partial r}{(t, r)}
  \right\}e^{i\frac{r^2}{2t}} \\
&+i \left\{ \frac{a}{2t}(1+r H(x, y))+\frac{r}{t}\frac{\partial a}{\partial r}{(t, r)}+ \frac{\partial a}{\partial t}{(t, r)}
   \right\} e^{i\frac{r^2}{2t}}.
\end{align*}
Recall that $a=a(t, r)$ satisfies the transport equation (See e.g. \cite[\S7.3]{Go}): 
$$\frac{\partial a}{\partial t}+\nabla a \cdot \nabla S+\frac{a}{2}\triangle S 
=\frac{\partial a}{\partial t}+\frac{r}{t} \frac{\partial a}{\partial r}+\frac{a}{2t}(1+r H(x, y)) =0. 
$$
It follows that 
\begin{align*}
\left(i\frac{\partial}{\partial t}+\frac{1}{2}\triangle_x \right) \widehat K(t, x, y)
&=\left\{ \frac{1}{2}\frac{\partial ^2 a}{\partial r^2}{(t, r)} +\frac{H(x, y)}{2} \frac{\partial a}{\partial r}{(t, r)}
  \right\} e^{i\frac{r^2}{2t}} \\
&=\left\{\frac{1}{2}\triangle_x a(t, r) \right\} e^{i\frac{r^2}{2t}}. 
\end{align*}
%%%%%%%%%%%%%%%%%%%%%%%%%%%%%%%%%%%%%%%%%
\begin{lemma}
$\left[ \frac{1}{2}\frac{\partial ^2 a}{\partial r^2}{(t, r)} +\frac{H(x)}{2} \frac{\partial a}{\partial r}{(t, r)}
  \right]|_{r=0}=\frac{1}{2} \triangle_x {a(t, x, y)}|_{x=y} = \frac{1}{2t^{n/2}}\cdot \frac{R(y)}{6}$
\end{lemma}
%%%%%%%%%%%%%%%%%%%%%%%%%%%%%%%%%%%%%%%%
\begin{proof}
Taking Riemann normal coordinates at $y$, it is known \cite{K-N} that 
$$
\begin{cases}
&g_{ij}(x)=\delta_{ij}-\frac{1}{3}R_{ikjl}(y) x^k x^l+O(|x|^3),  \vspace{2mm}\\
&\sqrt{\det g(x)}=1-\frac{1}{6}R_{ij}(y) x^i x^j+O(|x^3|). 
\end{cases}
$$
Here $R_{ikjl}$ and $R_{ij}$ denote the purely covariant version of Riemann curvature tensor and 
the Ricci curvature tensor. From Lemma 2.1, 
$$
a(t, x, y)=\sqrt{V(t, x, y)}=t^{-n/2} \{ \det(g_{i j}(x)) \}^{-1/4}. 
$$
Remarking that $\triangle=\sum_{i}\frac{\partial^2}{\partial (x^i) ^2}$ at $x=y$,  
\begin{align*}
\triangle_x {a(t, x, y)}|_{x=y}&=t^{-n/2} \triangle_x {\{ \det(g(x)) \}^{-1/4}}|_{x=y}\\
                                   &=t^{-n/2} \triangle_x \left(1+\frac{1}{12}R_{ij}(y) x^i x^j+o(|x^2|)\right)|_{x=y} \\
                                   &=t^{-n/2}\cdot \frac{R(y)}{6}   
\end{align*}
as desired. 
\end{proof}
%%%%%%%%%  Stationary phase  %%%%%%%%%%%%%%%%%%%%%%%%%%%%%%%%%%%%
\par
For $\widehat K(t, x, y) \equiv \chi(d(x, y))K(t, x, y)=\chi(d(x, y)) a(t, x, y) e^{i \frac{d(x, y)^2}{2t}}$, 
we obtain  
\begin{align*}
&\left(i\frac{\partial}{\partial t}+\frac{1}{2}\triangle_x \right)
(\chi K(t,x,y)) \\
&=\chi \left\{\frac{1}{2}\triangle_x a(t, x, y) \right\} e^{i\frac{d^2}{2t}}+\frac{1}{2}(\triangle_x \chi)K(t, x, y)
+\frac{1}{2}\left(2\nabla_x \chi \cdot \nabla_x K(t,x,y)\right) 
\end{align*}
where $\nabla_x$ is gradient with respect to $x$. 
Seeing this, we define the error integral $E_{\chi_1}(t)$ and $E_{\chi_2}(t)$ by 
$$
\begin{cases} 
&E_{\chi_1}(t) f(x) \equiv \frac{1}{(2\pi i)^{n/2}}\int_{M}
\big[ \chi \left\{\frac{1}{2}\triangle_x a(t, x, y) \right\} 
+\frac{1}{2}(\triangle_x \chi) a(t, x, y)  \big] e^{i\frac{d^2}{2t}}  f(y)\; dy, \\
&E_{\chi_2}(t) f(x) \equiv \frac{1}{(2\pi i)^{n/2}}\int_{M} 
\nabla_x \chi \cdot \nabla_x K(t,x,y) f(y)\; dy. 
\end{cases}
$$
From Lemma 2.2, $\big[\chi \left\{\frac{1}{2}\triangle_x a(t, x, y) \right\}
+\frac{1}{2}(\triangle_x \chi) a(t, x, y)  \big]|_{d=0}=\frac{R(x)}{12t^{n/2}}$. 
In order to estimate $U_{\chi}(t)$, $E_{\chi_1}(t)$ and $E_{\chi_2}(t)$, 
%%%%%%%%%%%%%%%%%%%%%%%%%%%%%%%%%%%%%%%%%%%%%%%%%%%%%%%%%%%%%%%
we state the method of stationary phase where $S(t,x,y)$ is a quadratic form,  
which is convenient here (See \cite[Lemma 7.7.3]{Ho}). 
\begin{lemma}
Let $A$ be a symmetric non-degenerate matrix with ${\rm Im} A \geqq 0$. 
Then we have for every integer $k>0$ and integer $s>n/2$
\begin{align*}
\Big| \int_{\mathbf{R}^n} u(x) e^{\frac{i <Ax, x>}{2t}}\; dx
-(\det(& A/2\pi i t))^{-\frac{1}{2}}
\sum_{j=0}^{k-1}(-it/2)^j \langle A^{-1}D, D \rangle^j u(0)/j! \Big| 
\\
&\leqq C_k (\Vert A^{-1} \Vert t)^{n/2+k} 
\sum_{|\alpha|\leqq 2k+s} \Vert D^{\alpha}u \Vert_{L^2(\mathbf{R}^n)}, 
\quad \text{for}\ \ u(x)\in \mathcal{S}(\mathbf{R}^n) .
\end{align*}
\end{lemma}
The right hand side in the above lemma is just the Sobolev norm :   
$$\Vert \cdot \Vert_{H^{2k+s}(\Omega)}
=\sum_{|\alpha|\leqq 2k+s} \Vert D^{\alpha} \cdot \Vert_{L^2 (\Omega)}.$$   
%%%%%%%%%%%%%%%%%%%%%%%%%%%%%%%%%%%%%%%%%%%%%%%%%%%%%%%%%%%%%%%%%%
Letting $A=I$ (unit matrix), we obtain stationary phase lemma 
in the polar coordinate system of ${\mathbf{R}}^n$: 
\begin{corollary} Let $\chi(r)\in C_0^{\infty}(\mathbf{R}) $ be 
the bump function 
with compact support contained in $|r|<\underline{\it d}$. 
Then we have for every integer $k>0$ and integer $s>n/2$
\begin{align*}
\Big| \int_{S^{n-1}} \int_{\mathbf{R}} \chi(r) \ u(r,\theta) 
e^{\frac{i r^2}{2t}}g_{st}(r, \theta)\; dr &d\theta
-( 2\pi i t)^{n/2}\sum_{j=0}^{k-1}(it\triangle_{flat}/2)^j  u(0)/j! \Big| 
\\
&\leqq \tilde C_k t^{n/2+k} \Vert \chi u \Vert_{H^{2k+s}(\Omega_{\underline{\it d}})}
\quad \text{for}\ \ u(r, \theta)\in {C}^{\infty}(\mathbf{R}^n, \mathbf{C}), 
\end{align*}
where $g_{st}(r, \theta) d\theta$ denotes the spherical volume form on $S^{n-1}(r)$ and 
$\Omega_{\underbar{\it d}}=\{ x \in{\mathbf{R}^n}\ |\ r=|x|<\underbar{\it d} \}$.  
\end{corollary}
%%%%%%%%%%%%%%%%%%%%%%%%%%%%%%%%%%%%%%%%%%%%%%%%%%%%%%%%%%%%%%
To obtain the criteria for compact Riemannian manifolds, we recall Gauss's lemma \cite{K-N}  
which asserts that the line element for  geodesic polar coordinates on $M$ is given by 
$$
ds^2=dr^2+g_{i j}(r, \theta) d\theta_i d\theta_j . 
$$
In particular, letting $r\rightarrow 0_{+}$, we know: 
$$
\frac{g(r, \theta)}{r^{n-1}}=\frac{\sqrt{\det g_{ij}(r, \theta)}}{r^{n-1}} \rightarrow 1.
$$
So 
\begin{equation}
g_{st}(r, \theta)/g(r, \theta) \rightarrow 1. 
\end{equation}
By Corollary 2.4 and putting $k=1$ and $k=2$, we have 
%%%%%%%%%%%%%%%%%%%%%%%%%%%%%%%%%%%%%%%%%%%%%%%%%%%%%%%%%%%%%%
\begin{proposition} Let $\alpha=2+\frac{1}{2}[\frac{n+2}{2}]$. For $d<\underline{\it d}$ and $x\in M$, 
\begin{align*} 
&\Big| U_{\chi} (t) f(x)- f(x) \Big| \leqq C t \Vert (-\triangle+1)^{\alpha-1} f \Vert_{L^2(M)}, \\ 
&\Big| E_{\chi_1}(t) f(x)-\frac{R(x)}{12}f(x) \Big|  \leqq C t \Vert (-\triangle+1)^{\alpha-1} f \Vert_{L^2(M)}  \\
& and \\
&\Big| E_{\chi_2}(t) f(x) \Big|  \leqq C t \Vert (-\triangle+1)^{\alpha} f \Vert_{L^2(M)}
\quad \text{for}\ \ f(x)\in C^{\infty}(M). 
\end{align*}
%where $\Omega_{R_x}=\{ y\in S^2\ |\ d(x,y)<R \}$.
\end{proposition}
%%%%%%%%%%%%%%%%%%%%%%%%%%%%%%%%%%%%%%%%%%%%%%%%%%%%%%%%%%%%%%%
\begin{proof} Take $x$-centered geodesic polar coordinate and $\Omega_{R_x}=\{ y\in S^2\ |\ d(x,y)<R \}$. 
\begin{align*} 
\Big| U_{\chi} (t) f(x)- f(x) \Big| &= 
\Big|\frac{1}{(2\pi i)^{n/2}}\int_{M} \chi(d(x,y)) 
\sqrt{V(t, x, y)} e^{iS(t,\ x,\ y)}f(y)\;g_r(r, \theta) dr d\theta -f(x) \Big|  \\
& = \Big|\frac{1}{(2\pi i)^{n/2}}\int_{M} \chi(d(x,y)) 
\sqrt{V(t, x, y)} e^{iS(t,\ x,\ y)}f(y)\;\frac{g_r(r, \theta)}{g_{st}(r, \theta)} 
g_{st}(r, \theta) dr d\theta  -f(x) \Big|
\end{align*}
Note that 
$\frac{g_{r}(r, \theta)}{g_{st}(r,\theta)}|_{r=0}=1$ and $\chi(d(x,y)) \sqrt{V(t, x, y)}\;|_{d=0}=t^{-n/2}$ 
from Lemma 2.1 and $g_{ij}(x)|_{d=0}=\delta_{ij}$. By Corollary 2.4 and putting $k=1$, 
$$ \Big| U_{\chi} (t) f(x)- f(x) \Big| \leqq c_1 t \Vert \chi f \Vert_{H^{2+s}_{flat}(\Omega_{R_x})}. $$
Similarly $\big[ \chi \left\{\frac{1}{2}\triangle_x a(t, x, y) \right\} 
+\frac{1}{2}(\triangle_x \chi) a(t, x, y)  \big]|_{r=0}=R(x)/12t^{n/2}$, it follows that 
$$ \Big| E_{\chi_1} (t) f(x)-R(x)/12 \Big| \leqq c_2 t \Vert \chi f \Vert_{H^{2+s}_{flat}(\Omega_{R_x})}. $$    
To estimate $E_{\chi_2}(t)$ we need to put $k=2$, since $\nabla_x K(t, x, y)$ gives higher-order singular $1/t$ terms. $\nabla_x \chi=0$ on the neighborhood of $d=0$ and so Corollary 2.4 can be applied 
to $E_{\chi_2}(t)$:    
$$ \Big| E_{\chi_2} (t) f(x) \Big| \leqq c_3 t \Vert \chi f \Vert_{H^{4+s}_{flat}(\Omega_{R_x})}. $$ 
%%%%%%%%%%%%%%%%%%%%%%%%%%%%%%%%%%%%%%%%%%%%%%%%%%%%%%%%%%%%%%%%
Thus we only have to use    
$$ 
\Vert \chi f \Vert_{H^N_{flat}(\Omega_{R_x})} \leqq c_4 \Vert (-\triangle+1)^{N/2} f \Vert_{L^2(M)} \\ 
$$ 
on local charts (See e.g. \cite{Au}).  We mention shortly this inequality for the reader's convenience.  
\par 
Take one "atlas" $\mathcal{A}$. Making the change of variables $y=T(x)$ and 
using $dy= |\det T|' dx \leqq \epsilon dx$, 
$$
\Vert \chi f \Vert_{H^N_{flat}(\Omega_x)} 
$$
are equivalent under changing coordinates. 
Furthermore, comaring with flat and manifold metric and using $g_{flat} \sim g$ on small local charts,  
$$
\Vert \chi f \Vert_{H^N_{flat}(\Omega_x)} \leqq c_5 \Vert \chi f \Vert_{H^N(M)}.
$$
Let $\phi_i$ be a partition of unity associated to $\mathcal{A}$. 
Recall that $\chi$ is said to be $C^{\infty}$ if $\chi\circ x_i^{-1}\in C^{\infty}$, we find  
%%%%%%%%%%%%%
$$
\Vert (\phi_i \chi f)\circ x_i^{-1}\Vert_{H^N(M)}
=\Vert [ \chi \circ x_{i}^{-1}](\phi_i f )\circ x_{i}^{-1} \Vert_{{H^N}(M)}
\leqq C_i \Vert (\phi_{i}f) \circ x_{i}^{-1} \Vert_{H^N(M)} 
$$ 
and summing this equation on $i$ shows 
$\Vert \chi f \Vert_{H^N(M)} \leqq 
c_6 \Vert f \Vert_{H^N{(M)}}$ 
holds with $c_6 \equiv max_{i}C_{i}$. 
Summarizing the calculations, 
\begin{equation}
\Vert \chi f \Vert_{H^N_{flat}(\Omega_{R_x})} \leqq c_6 \Vert f \Vert_{H^N(M)}.  
\end{equation}
We apply G\r{a}rding inequality of elliptic operators to (3), 
$$
\Vert \chi f \Vert_{H^N_{flat}(\Omega_{R_x})} \leqq c_7 \Vert f \Vert_{H^N(M)}
\leqq C \Vert (-\triangle+1)^{N/2} f \Vert_{L^2(M)}.  
$$
\end{proof}
%%%%%%%%%%%%
\begin{proposition}
Let $(M, g)$ be a compact Riemannian manifold. For $\alpha=2+\frac{1}{2}[\frac{n+2}{2}]$
\begin{equation}
\Vert \{ U_{\chi}(t) f(x) 
-e^{\frac{it}{2}(\triangle-\frac{R(x)}{6})}\} f(x) \Vert_{L^2}
\leqq  \frac{C t^2}{2} \Vert (-\triangle+1)^{\alpha} f(x) \Vert_{L^2}.
\end{equation} 
%\Vert  \int_{0}^t  e^{\frac{-is}{2}(\triangle-\frac{R(x)}{6})}  E(s) f(x) ds  \Vert_{L^2}
\begin{proof}
 For $f(x) \in C^{\infty}(M)$ and $E(t)=E_{\chi_1}(t)+E_{\chi_2}(t)$  
$$
\left(i\frac{\partial}{\partial t}+\frac{1}{2}\triangle_x \right)U_{\chi}(t) f(x)
={E(t)} f(x).  
$$
and so 
$$
\left(i\frac{\partial}{\partial t}+\frac{1}{2}\triangle_x-\frac{R(x)}{12} \right)U_{\chi}(t) f(x)
=\left( {E(t)} -\frac{R(x)}{12}U_{\chi}(t) \right) f(x).  
$$
Letting $ \widehat{E(t)}= {E(t)} -\frac{R(x)}{12}U_{\chi}(t)$,  
$$
U_{\chi}(t) f(x) 
=e^{\frac{it}{2}(\triangle-\frac{R(x)}{6})}
\left(1+ \int_{0}^t e^{\frac{-is}{2}(\triangle-\frac{R(x)}{6})} \widehat{E(s)} ds \right) f(x).   
$$
From Proposition 2.5 
\begin{align*}
|\widehat{E(t)} f(x)| &\leqq \Big|\left({E(t)} -\frac{R(x)}{12} \right)f(x)\Big| 
+  \Big|\left(\frac{R(x)}{12} {U_{\chi}(t)} -\frac{R(x)}{12} \right)f(x)\Big| \\
           &\leqq  C t \Vert (-\triangle+1)^{\alpha} f \Vert_{L^2(M)}.
\end{align*}
It follows 
\begin{align*}
\Vert  \int_{0}^t  e^{\frac{-is}{2}(\triangle-\frac{R(x)}{6})}  \widehat{E(s)} f(x) ds  \Vert_{L^2}
&\leqq  \int_{0}^t \Vert e^{\frac{-is}{2}(\triangle-\frac{R(x)}{6})} \widehat{E(s)} f(x) \Vert_{L^2} ds \\
&=  \int_{0}^t \Vert \widehat{E(s)} f(x) \Vert_{L^2} ds \\
&\leqq  \int_{0}^t  \tilde{C} s \Vert (-\triangle+1)^{\alpha} f(x) \Vert_{L^2} ds \\
&\leqq  \frac{\tilde{C} t^2}{2} \Vert (-\triangle+1)^{\alpha} f(x) \Vert_{L^2}
\end{align*} 
as desired. 
\end{proof}
\end{proposition}
%%%%%%%%%%%%%%%%%%%%%%%%%%%%%%%%%%%%%%%%%%%%%%%%%%%%%%%%%%%%%%%%%%%
%%%%%%%%%%%%%%%%%% section 3 %%%%%%%%%%%%%%%%%%%%%%%%%%%%%%%%%%%%%%
%%%%%%%%%%%%%%%%%%%%%%%%%%%%%%%%%%%%%%%%%%%%%%%%%%%%%%%%%%%%%%%%%%%
\section{Feynman path integral  for low energy functions on rank $1$ locally symmetric Riemannian manifolds}
%%%%%%%%%%%%%%%%%%%%%%%%%%%%%%%%%%%%%%%%%%%%%%%%%%%%%%%%%%%%%%%%%
The purpose of this section is to show 
the products of $U_\chi$'s 
converge uniformly for low energy functions in $L^2$. 
From now on, till the end of \S 4, we only consider rank $1$ locally symmetric Riemannian manifolds. 
These manifolds possess nice geometric properties; in particular, they 
are two-point homogeneous spaces. That is the isometry group on $(M, g)$ is transitive on 
the set of all equidisitant point pairs. Such a situation allows us that the commutator 
$[\triangle, d(x, y)]=0$ locally and the scalar curvature $R$ is constant (See e.g. \cite{SH}). 
Moreover $\sqrt{V}$ depends only on the distance 
for $d<\underbar{\it d}$ and  is considered as the density of paths connecting $x$ and $y$ 
(See e.g. \cite{Wa}). 
We abbreviate $U_{\chi}(t)-e^{\frac{it}{2}(\triangle-\frac{R(x)}{6})}$ 
to $\tilde{E}(t)$ in the following sentences. 
\setcounter{section}{3}
\setcounter{theorem}{0}
\begin{proposition}[Time slicing products and energy limits]
Let $(M, g)$ be a  compact, oriented,  rank $1$ locally symmetric Riemannian manifold. 
For small $\varepsilon>0$, we have 
$$
\lim_{N\rightarrow \infty} \Vert [\{U_\chi(t/N)\}^N 
-e^{\frac{it}{2}(\triangle-\frac{R}{6})}]\rho({N^{1/\alpha-\varepsilon}}) \Vert_{L^2}=0. 
$$
\end{proposition}
%%%%%%%%%%%%%%%%%%%%%%%%%%%%%%%%%%%%%%%%%%%%%%%%%%%%%%%%%%%%%%%%%%%%
\begin{proof} 
From (4)  
$$
\Vert  \tilde{E} (t)  \Vert_{L^2}
\leqq  \frac{C t^2}{2} \Vert (-\triangle+1)^{\alpha} f(x) \Vert_{L^2}.
$$
Let $\hat{H}=\triangle-\frac{R}{6}$.  
If $M$ is a rank $1$ locally symmetric Riemannian manifold, $R(x)$ is a constant function and  
$\tilde E(t) \hat{H} =\hat{H} \tilde E(t) $. Consequently we have 
\begin{align*}
\Vert \underbrace {e^{\frac{it \hat{H}}{2N}} e^{\frac{it \hat{H}}{2N}} \cdots 
e^{\frac{it \hat{H}}{2N}}}_{N-k\ \mbox{times}} 
\underbrace {\tilde E(t/N) \tilde E(t/N) \cdots \tilde E(t/N)}_{k\ \mbox{times}} f(x)\Vert_{L^2} 
\leqq \Big(\frac{{C}}{2}\Big)^k \Big(\frac{t}{N}\Big)^{2k} \Vert (-\triangle+1)^{k\alpha} f(x) \Vert_{L^2}. 
\end{align*}
%%%%%%%%%%%%%%
The binomical coefficients bounds $\begin{pmatrix} N \\ k \end{pmatrix}\frac{1}{N^k}<\frac{1}{k!}$ 
yields the following estimates 
\begin{align*}
\Vert \{e^{it\hat{H}/2}-U_\chi(t/N)^n\}f(x) \Vert_{L^2} 
&=\Vert \left[e^{it\triangle/2}-\{e^{it\hat{H}/2N}(1+\tilde E(t/N))\}^N \right] f(x) \Vert_{L^2} \\  
&\leqq \sum_{k=1}^N 
\begin{pmatrix} N \\ k \end{pmatrix} 
\Vert \{e^{i(N-k)t \hat{H}/2N} \tilde E(t/N)^{k}\}f(x)  \Vert_{L^2} \\
&\leqq \sum_{k=1}^N 
\begin{pmatrix} N \\ k \end{pmatrix} 
\Big(\frac{{C}}{2}\Big)^k \Big(\frac{t}{N}\Big)^{2k} 
\Vert (-\triangle+1)^{k\alpha} f(x) \Vert_{L^2} \\
&\leqq \sum_{k=1}^N 
\frac{1}{k!}
\Big(\frac{{C}t^2}{2N} \Big)^k 
\Vert (-\triangle+1)^{k\alpha} f(x) \Vert_{L^2}. 
\end{align*}
By using $\Vert (-\triangle+1)^{k\alpha} \rho(E) f(x) 
\Vert_{L^2}\leqq (E+1)^{k\alpha} \Vert f(x) \Vert_{L^2}$,    
\begin{align*}
\Vert \{e^{it\hat{H}/2}-U_\chi(t/N)^N\} \rho(E) f(x) \Vert_{L^2} 
&\leqq \sum_{k=1}^N 
\frac{1}{k!}
\Big\{\frac{C(E+1)^\alpha t^2}{2N}\Big\}^k \Vert f(x) \Vert_{L^2} \\ 
&\leqq \Big[\exp \Big\{ \frac{C(E+1)^\alpha t^2}{2N} \Big\}-1 \Big] \Vert f(x) \Vert_{L^2} \\
&\leqq \frac{C_2(E+1)^\alpha t^2}{2N} \Vert f(x) \Vert_{L^2}. 
\end{align*}
Thus for small $\epsilon >0$, 
$$
\lim_{N\rightarrow \infty} \Vert [\{U_\chi(t/N)\}^N  
-e^{\frac{it\triangle}{2}}] 
\rho({N^{1/\alpha-\varepsilon}}) \Vert_{L^2} \leqq 
\lim_{N\rightarrow \infty} \frac{C_2(N^{1/\alpha-\varepsilon}+1)^\alpha t^2}{2N}=0. $$ 
\end{proof} 
\par 
\begin{remark}
We note that $s\hspace{-1.5mm}\lim\limits_{E\rightarrow \infty} 
e^{\frac{it}{2}(\triangle-\frac{R}{6})} \rho(E) f(x)
=e^{\frac{it}{2}(\triangle-\frac{R}{6})} f(x)$, so 
$$
s\hspace{-1.5mm}\lim_{N\rightarrow \infty} \{U_\chi(t/N)\}^N 
\rho({N^{1/\alpha-\epsilon}}) f(x) 
=e^{\frac{it}{2}(\triangle-\frac{R}{6})} f(x) \quad\mbox{for}\ \forall\; f(x)\in L^2(M).  
$$
In $\S 4$, we show the stronger result by substituing $\rho({N}) $ for $\rho({N^{1/\alpha-\epsilon}}) $. 
\end{remark}
\par
\begin{remark}
Some Trotter-Kato formulas for Feynman's operational caluculus contain  
infinite many spectral projectors, however we used a spectral projector once only(See \cite{I-T}).  
\end{remark}
%%%%%%%%%%%%%%%%%%%%%%%%%%%%%%%%%%%%%%%%%%%%%%%%%%%%%%%%%%%%%%%%%
%%%%%%%%%%%%%%%%%% section 4 %%%%%%%%%%%%%%%%%%%%%%%%%%%%%%%%%%%%%%
%%%%%%%%%%%%%%%%%%%%%%%%%%%%%%%%%%%%%%%%%%%%%%%%%%%%%%%%%%%%%%%%%
\setcounter{section}{3}
\setcounter{theorem}{1}
\section{Strong limits for high energy functions}
%%%%%%%%%%%%%%%%%%%%%%%%%%%%%%%%%%%%%%%%%%%%%%%%%%%%%%%%%%%%%%%%%
In this section, we have 
the strong but not uniform convergence of time slicing products. 
To do this, we introduce the $L^2$ estimates known as H\"ormander and Maslov's 
theorem (See e.g. \cite[Theorem 2.1.1]{S} for more details). 
\begin{lemma} Let $a\in C_0({\mathbf{R}^n})$ and assume that 
$\Phi \in C^\infty$ satisfies $|\nabla \Phi|\geqq c >0$ on supp\ $a$. 
Then for all $\lambda>1$, 
$$
\Big|\int_{\mathbf{R}^n} a(x)e^{i \lambda \Phi(x)}\; dx \Big| 
\leqq C_{N} \lambda^{-N}, \quad N=1, 2, \cdots
$$ 
where $C_N$ depends only on $c$ if $\Phi$ and $a$ belong to a bounded 
subset of $C^{\infty}$ and $a$ is supported in a fixed compact set. 
\end{lemma}
\begin{proof}
Given $x_0 \in \mbox{supp}\; a$ there is a direction $\nu \in S^{n-1}$ 
such that $|\nu \cdot \nabla \Phi|\geqq \frac{c}{2}$ on some 
ball centered at $x_0$. Thus, by compactness, we can choose a partition 
of unity $\alpha_j \in C_0^{\infty}$ consisting of a finite number 
of terms and corresponding unit vectors $\nu_{j}$ 
such that $\sum \alpha_j(x)=1$ on $\mbox{supp}\ a$ and 
$|\nu_j \cdot \nabla \Phi|\geqq \frac{c}{2}$ on $\mbox{supp}\ \alpha_j$. 
If we set $a_j(x)=\alpha_{j}(x)a(x)$, it suffices to prove that 
for each $j$ 
$$
\Big|\int_{\mathbf{R}^n} a_j(x)e^{i \lambda \Phi(x)}\; dx \Big| 
\leqq C_{N} \lambda^{-N}, \quad N=1, 2, \cdots
$$ 
After possibly changing coordinates we may assume that 
$\nu_j=(1, 0, \ldots, 0)$ which means that 
$|\partial \Phi/ \partial x_1| \geqq c/2$ on supp $a_j$. If we let 
$$
L(x, D)=\frac{1}{i\lambda \partial \Phi/ \partial x_1}
\frac{\partial}{\partial x_1}, 
$$
then 
$L(x, D)e^{i\lambda \Phi(x)}=e^{i\lambda \Phi(x)}$. 
Consequently, if 
$
L^* = L^*(x, D)=\frac{\partial}{\partial x_1}
\frac{1}{i\lambda \partial \Phi/ \partial x_1}
$
is the adjoint, then 
$$
\int_{\mathbf{R}^n} a_j(x)e^{i \lambda \Phi(x)}\; dx = 
\int_{\mathbf{R}^n} (L^*)^N a_j(x)e^{i \lambda \Phi(x)}\; dx.  
$$
Since our assumptions imply that $(L^*)^N a_j=O(\lambda^{-N})$, 
the results follows. 
\end{proof}
%%%%%%%%%%%%%%%%%%%%%%%%%%%%%%%%%%%%%%%%%%%%%%%%%%%%%%%%%%%%%%%%%%%
\begin{lemma}
Suppose that $\phi(x, y)$ is a real $C^{\infty}$ function satisfying 
the non-degeneracy condition 
$$
\det \left(\frac{\partial^2 \phi}{{\partial x_j}{\partial y_k}}\right)\not =0
$$
on the support $a(x, y)\in C_0^{\infty}(\mathbf{R}^n \times \mathbf{R}^n)$. 
Then for $t>0$, 
$$
\Vert \int_{\mathbf{R}^n} e^{i\frac{\phi(x, y)}{2t}} a(x, y) f(y)\; dy \Vert_{L^2({\mathbf{R}}^n)}
\leqq C t^{n/2} \Vert f \Vert_{L^2({\mathbf{R}}^n)}. 
$$
where $C$ is indep. of $t$ and $f(x)$. 
\end{lemma}
\begin{proof}
We note that 
$$ 
\nabla_x[\phi(x, y)-\phi(x, z)]
=\left(\frac{\partial^2 \phi(x, y)}{{\partial x_j}{\partial y_k}}\right)(y-z)
+O(|y-z|^2). 
$$
By using a smooth partition of unity we can decompose 
$a(x, y)$ into a finite number of pieces each of which 
has the property that 
\begin{equation}
|\nabla[\phi(x, y) - \phi(x, z)]|\geqq c|y-z| \quad \mbox{on supp}\ a, 
\end{equation}
for some $c>0$. 
\par 
To use this we notice that 
\begin{equation}
\Vert T_{t} f\Vert_2^2=\int\int K_t(y, z) f(y)\overline{f(z)}\ dy\; dz, 
\end{equation}
where 
$$
K_{t}(y, z)=\int_{{\mathbf{R}}^n} e^{\frac{i}{t}[\phi(x, y) - \phi(x, z)]} a(x, y) \overline{a(x, z)} dx. 
$$
However, (5) and Lemma 4.1 imply that 
$$
|K_{t}(y, z)|\leqq C_N(1+\frac{1}{t}|y-z|)^{-N} \quad \mbox{for}\ \ \forall N. 
$$
By appling Schur test, the operator with kernel $K_t$ sends 
$L^2$ into itself with norm $O(t^n)$. This along with (6) yields 
$$
\Vert T_{t} f\Vert^2_{L^2({\mathbf{R}}^n)} 
\leqq C t^{n} \Vert f \Vert^2_{L^2({\mathbf{R}}^n)}, 
$$
as desired. 
\end{proof}
\begin{lemma}
\begin{align*}
\Vert (\int_{0}^t e^{\frac{-is \widehat{H}}{2}} E_{\chi}(s) f(x) ds) \Vert_{L^2}
\leqq 
C_1 t \Vert  f(x) \Vert_{L^2}  
+ {C_2 t^2} \Vert (-\triangle+1) f(x) \Vert_{L^2}
\end{align*} 
\end{lemma}
\begin{proof}
We shall use the partition of unity $\{\phi_i\}$ on $M$ 
with very small support $\mbox{diam}\ \phi_i <\underline{d}$.
\par
If $d(\mbox{supp}(\phi_i) , \mbox{supp}(\phi_j))>\underline{d}+2\epsilon$, 
$$
\phi_j(x) \{ E_{\chi}(t) (\phi_i(y) f(y))\}(x)=0. 
$$
So we may assume $\phi_i$ and $\phi_j$ are contained in one local chart.   
The same calculation for $E_\chi$ on local charts as Lemma 4.2 implies 
$$
\Vert T_{t, i, j, k, l} f\Vert_2^2=\int\int K_{t, i, j, k, l}(y, z)
\{g^{1/2}(y)f(y)\} \overline{\{g^{1/2}(z)f(z)\}}\ dy\; dz,   
$$
where 
$$
K_{t, i, j, k, l}(y, z)=\int_{{\mathbf{R}}^2} e^{\frac{i}{2t}[d^2(x, y) - d^2(x, z)]} \phi_i(x) \overline{\phi_j(x)} 
\phi_{k}(y) \overline{\phi_{l}(z)}
a(x, y) \overline{a(x, z)} dx. 
$$
%%%%%%%%%%%%%%%%%%%%%%%%%%%%%%%%%%%%%%%%%%%%%%%%%%%%%%%%%%%%%%%%%%%%%%%%%%%%%
We give a simple explanation of the boundedness of $T_{t, i, j, k, l}$. 
$M$ is compact and so $c_1<g(y)<c_2$. From Lemma 2.1, 
$$
\det _{ij} \left(\frac{\partial^2 d^2}{{\partial {x_i}}{\partial {y_j}}}\right)
\quad \text{for}\ 0\leqq d <\underline{d} .
$$
is also bounded.  
Applying Lemma 4.2, we have $ \Vert T_{t, i, j, k, l} f\Vert_2<C $.  
$i, j$'s are finite and we conclude 
%%%%%%%%%%%%%%%%%%%%%%%%%%%%%
\begin{equation}
\Vert E_{\chi_1}(t) f \Vert_{L^2} =\Vert \sum\limits_{i,j} \phi_i(x) 
\{E_{\chi_1}(t) \phi_j(y) f\}(x) \Vert_{L^2}<C_3\Vert f \Vert_{L^2}.  
\end{equation}
%%%%%%%%%%%%%%%%%%%%%%%%%%%%%%
For $E_{\chi_2}(t)$, we have 
\begin{align*}
&E_{\chi_2}(t)f(x) \\
&= \frac{1}{(2\pi i)^{n/2}}\int_{M} 
\nabla_x \chi \cdot \nabla_x K(t,x,y) f(y)\; dy \\ 
&= \frac{1}{(2\pi i)^{n/2}}\int_{M} 
\nabla_x \chi \cdot \nabla_x \{a(t, x, y)e^{\frac{id(x, y)^2}{2t}}\} f(y)\; dy \\
&= \frac{1}{(2\pi i)^{n/2}}\int_{M} 
\{\nabla_x \chi \cdot \nabla_x a(t, x, y)\}  e^{\frac{id(x, y)^2}{2t}} f(y)\; dy 
+\frac{1}{(2\pi i)^{n/2}}\int_{M} 
a(t, x, y) \nabla_x \chi \cdot \nabla_x e^{\frac{id(x, y)^2}{2t}}  f(y)\; dy \\
&=\frac{1}{(2\pi i)^{n/2}}\int_{M} 
\{\nabla_x \chi \cdot \nabla_x a(t, x, y)\}  e^{\frac{id(x, y)^2}{2t}} f(y)\; dy 
+\frac{1}{(2\pi i)^{n/2}}\int_{M} 
a(t, x, y) \frac{\partial \chi}{\partial d} 
\frac{\partial \{e^{\frac{id(x, y)^2}{2t}}\} }{\partial d}   f(y)\; dy.  
\end{align*}
Since the inner product of $\nabla_x g(d(x,y)) \cdot \nabla_x h(d(x, y))
=\frac{\partial g(d(x,y))}{\partial d} \frac{\partial h(d(x,y))}{\partial d}$ from Gauss's lemma about  
normal charts. The first term is bounded by using the same method of $E_{\chi_1}$. 
We estimate the second term. Letting $b(x, y)=t^{n/2}a(t, x, y)\frac{\partial \chi}{\partial d}$  
%%%%%%%%%%%%%%%%%%%%%%%%%%%%%%
\begin{align*}
&\frac{1}{(2\pi i)^{n/2}}\int_{M} 
a(t, x, y) \frac{\partial \chi}{\partial d} 
\frac{\partial \{e^{\frac{id(x, y)^2}{2t}}\} }{\partial d}   f(y)\; dy \\
&=\frac{1}{(2\pi i t)^{n/2}}\int_{M} 
b(x, y) \frac{id(x, y)}{t} e^{\frac{id(x, y)^2}{2t}} f(y)\; dy \\
&=\frac{1}{(2\pi i t)^{n/2}}\int_{M} 
\frac{b(x, y)}{d(x, y)} \frac{id(x, y)^2}{t} e^{\frac{id(x, y)^2}{2t}} f(y)\; dy \\
&=\frac{n-2}{(2\pi i t)^{n/2}}\int_{M} 
\frac{b(x, y)}{d(x, y)}e^{\frac{id(x, y)^2}{2t}} f(y)\; dy
+\frac{\partial}{\partial t}\left(\frac{1}{\pi i (2\pi i t)^{n/2-1}}\int_{M}\frac{b(x, y)}{d(x, y)}  e^{\frac{id(x, y)^2}{2t}}    f(y) \; dy \right) \\
&=\widetilde{E_{\chi_2}(t)} f(x)+\widehat{E_{\chi_2}(t)} f(x).
\end{align*}
$\frac{b(x, y)}{d(x, y)}=t^{n/2}\frac{a(t, x, y)}{d(x, y)}\frac{\partial \chi(d(x, y))}{\partial d(x, y)}$ is bounded. 
So we have 
\begin{equation} 
\Vert \widetilde{E_{\chi_2}(t)} f(x)\Vert_{L^2}\leqq C_4\Vert f \Vert_{L^2}
\end{equation} 
\begin{align*}
\Vert \int_0^t e^{\frac{-is \widehat{H}}{2}} \widehat{E_{\chi_2}(t)} f(x)\Vert  &
=\Vert \Big[e^{\frac{-is \widehat{H}}{2}} \left(\frac{1}{\pi i (2\pi i t)^{n/2-1}}\int_{M}\frac{b(x, y)}{d(x, y)}  
e^{\frac{id(x, y)^2}{2t}} f(y) \; dy \right) \Big]_0^t\Vert_{L^2} \\
& \quad +\Vert \left(\int_{0}^t \frac{-is \widehat{H}}{2} e^{\frac{-is \widehat{H}}{2}}  
\big\{ 
\frac{1}{\pi i (2\pi i t)^{n/2-1}}\int_{M}\frac{b(x, y)}{d(x, y)}  
e^{\frac{id(x, y)^2}{2t}} f(y) \; dy \big\} \; ds \right)\Vert_{L^2}
\\
& \leqq {C_5 t}\Vert f(x) \Vert_{L^2}+{C_6 t^2} \Vert (-\triangle+1) f(x) \Vert_{L^2} \tag{9}. 
\end{align*}
Summarizing (7), (8) and (9),  
\begin{align*}
\Vert (\int_{0}^t e^{\frac{-is \widehat{H}}{2}} 
\{ E_{\chi_1}(s)+E_{\chi_2}(s) \} f(x) ds) \Vert_{L^2}
\leqq 
C_1 t \Vert  f(x) \Vert_{L^2}  
+ {C_2 t^2} \Vert (-\triangle+1) f(x) \Vert_{L^2}
\end{align*} 
as desired. 
\end{proof}
%%%%%%%%%%%%%%%%%%%%%%%%%%%%%%%%%%%%%%%%%%%%%%%%%%%%%%%%%%%%%%%
It follows that $\{U_{\chi}(t/N)\}^N \rho(N)$ are uniformly bounded, 
and so the strong limit is obtained:  
\begin{main}[Time slicing products and the strong limits]
$$s\hspace{-1.5mm} \lim_{N\rightarrow \infty} \{U_{\chi}(t/N)\}^N \rho(N) f(x)
= e^{\frac{it}{2}(\triangle-\frac{R}{6})} f(x)
\quad \mbox{for}\; \forall \; f(x) \in L^2(M). $$
\end{main}
\begin{proof}
By Lemma 4.3, 
$\Vert U_{\chi}(t) \rho(E) f(x)\Vert \leqq \{1+C_1|t|+C_2 t^2 (E+1)\} \Vert f(x)\Vert_{L^2}$. 
Consequently 
$$
\Vert \{U_{\chi}(t/N)\}^N \rho(N)  f(x)\Vert 
\leqq {(1+C_1|t|/N+C_2(N+1)t^2/N^2)}^{N}\Vert f(x)\Vert_{L^2}<e^{C|t|}\Vert f(x)\Vert_{L^2}. 
$$ 
Letting $\widehat{H}=\triangle-\frac{R}{6}$, the estimates of Proposition 3.1 yields 
\begin{align*}
\lim_{N\rightarrow \infty} \Vert (e^{\frac{it\widehat{H}}{2}}-\{U_\chi(t/N)\}^N \rho(N) ) f(x) \Vert_{L^2}  
\leqq &
\lim_{N\rightarrow \infty} [\Vert e^{\frac{it\widehat{H}}{2}} 
(1- \rho({N^{1/\alpha-\varepsilon}}) ) f(x) \Vert_{L^2} \\
& + \Vert (e^{\frac{it\widehat{H}}{2}} -\{U_\chi(t/N)\}^N)\ \rho({N^{1/\alpha-\varepsilon}}) f(x) \Vert_{L^2} \\
&+\Vert \{U_\chi(t/N)\}^N\ (\rho(N) - \rho({N^{1/\alpha-\varepsilon}})) f(x) \Vert_{L^2}] \\
=&0. 
\end{align*}
\end{proof}
%%%%%%%%%%%%%%%%%%%%%%%%%%%%%%%%%%%%%%%%%%%%%%%%%%%%%%%%%%%%%%%%%
%%%%%%%%%%%%%%%%%% section 5 %%%%%%%%%%%%%%%%%%%%%%%%%%%%%%%%%%%%%%
%%%%%%%%%%%%%%%%%%%%%%%%%%%%%%%%%%%%%%%%%%%%%%%%%%%%%%%%%%%%%%%%%
\setcounter{section}{4}
\setcounter{theorem}{1}
\section{Some remarks}
\begin{remark}
Main theorem holds true even for two-point homogeneous spaces. For a torus, $R=0$ and 
$$s\hspace{-1.5mm} \lim_{N\rightarrow \infty} \{U_{\chi}(t/N)\}^N \rho(N) f(x)
= e^{\frac{it}{2}\triangle} f(x)
\quad \mbox{for}\; \forall \; f(x) \in L^2(M). $$ 
\end{remark}
%%%%%%%%%%%%%%%%%%%%%%%%%%%%%%%%%%%%%%%%%%%%%%%%%%%%%%%%%%%%%%%%%
\begin{remark}
Since $M$ is compact, we need not to use 
Cotlar-Stein lemma (See e.g. \cite[p.238]{Fu 3}). 
\end{remark}
%%%%%%%%%%%%%%%%%%%%%%%%%%%%%%%%%%%%%%%%%%%%%%%%%%%%%%%%%%%%%%%%%
\begin{remark}
Our estimates hold in Sobolev spaces (See \S 2), that is  
$$
s\hspace{-1.5mm} \lim_{N\rightarrow \infty} \{U_{\chi}(t/N)\}^N 
\rho({N}) f(x) 
=e^{\frac{it}{2}(\triangle-\frac{R}{6})} f(x) \quad \mbox{in}\ H^k(M).
$$ 
The Sobolev imbedding theorem yields the uniformly convergence:
$$
\lim_{N\rightarrow \infty}\sup\limits_{x \in M}| [\{U_{\chi}(t/N)\}^N 
-e^{\frac{it}{2}(\triangle-\frac{R}{6})}] \rho({N}) f(x)|=0 
\quad \mbox{for}\ f(x)\in H^k(M) 
$$ 
where $k>\frac{n}{2}=\frac{\mbox{dim} M}{2}$. 
\end{remark}
%%%%%%%%%%%%%%%%%%%%%%%%%%%%%%%%%%%%%%%%%%%%%%%%%%%%%%%%%%%%%%%%%%%%%%%
%%%%%%%%%%%%%%%%%%%%%%%%%%%%%%%%%%%%%%%%%%%%%%%%%%%%%%%
\begin{remark}
We employed the shortest paths on $M$. 
$U_{\chi}(t)$ is defined by the action integrals and van Vleck determinants. 
van Vleck determinants diverge at conjugate points, thus we ignore the long paths. 
\par
On $S^1$, however, we can take infinite many long paths 
for Fresnel integrable functions. On $M$, can one construct the analogy ?
\end{remark}
%%%%%%%%%%%%%%%%%%%%%%%%%%%%%%%%%%%%%%%%%%%%%%%%%%%%%%%%%%%%%%%%%%%
%%%%%%%%%%%%%%%%%% section 6 %%%%%%%%%%%%%%%%%%%%%%%%%%%%%%%%%%%%%%
%%%%%%%%%%%%%%%%%%%%%%%%%%%%%%%%%%%%%%%%%%%%%%%%%%%%%%%%%%%%%%%%%%%
\setcounter{section}{5}
\setcounter{theorem}{1}
\section{Conclusion}
%%%%%%%%%%%%%%%%%%%%%%%%%%%%%%%%%%%%%%%%%%%%%%%%%%%%%%%%%%%%%%%%%
\par
Simple WKB like formulas of Feynman integrations are discussed. 
Low energy approximations assure the unique classical paths. 
The quantum evolution is given by means of action integrals and van Vleck determinants. 
That is 
$\{U_{\chi}(t/N)\}^N \rho (N)$ converges to 
the modified Schr\"odinger operator in strong topology.  
We would like to mention about the case for general Riemannian manifolds in the future.  
%%%%%%%%%%%%%%%%%%%%%%%%%%%%%%%%%%%%%%%%%%%%%%%%%%%%%%%%%%%%%%%%%%%
%%%%%%%%%%%%%%%%%% section  (Acknowledge) %%%%%%%%%%%%%%%%%%%%%%%%%%%
%%%%%%%%%%%%%%%%%%%%%%%%%%%%%%%%%%%%%%%%%%%%%%%%%%%%%%%%%%%%%%%%%%%
\vspace{5mm}\par\noindent
{{\bf Acknowledgements}} 
\par\noindent
The author would like to thank the organizers 
of  ISQS23 for the kind invitation. 
The author also wishes to thank Professor A. Inoue and Professor N. Kumano-go 
for their valuable comments. 
%%%%%%%%%%%%%%%%%%%%%%%%%%%%%%%%%%%%%%%%%%%%%%%%%%%%%%%%%%%%%%%%%%%%%%%%%%%
%\section*{References}

%%%%%%%%%%%%%%%%%%%%%%%%%%%%%%%%%%%%%%%%%%%%%%%%%%%%%%%%%%%%%%%%%%%
%%%%%%%%%%%%%%%%%% Affiliation   %%%%%%%%%%%%%%%%%%%%%%%%%%%%%%%%%%
%%%%%%%%%%%%%%%%%%%%%%%%%%%%%%%%%%%%%%%%%%%%%%%%%%%%%%%%%%%%%%%%%%%
%%%%%%% 
\vspace{10mm}
\par 
Y.Miyanishi: Center for Mathematical Modeling and Data Science,  
%\par\quad\quad\quad\quad\quad\quad
Osaka University, 
\par\quad\quad\quad\quad\quad\quad
Machikaneyamacho 1-3, 
Toyonakashi 560-8531, Japan; 
\vspace{2mm}\par
e-mail: miyanishi@sigmath.es.osaka-u.ac.jp  
%%%%%%%%%%%%%%%%%%%%%%%%%%%%%%%%%%%%%%%%%%%%%%%%%%%%%%%%%%%%%%%%%% 
%%%%%%%%%%%%%%%%%%%%%%%%%%%%%%%%%%%%%%%%%%%%%%%%%%%%%%%%%%%%%%%%%%
\end{document}